\documentclass[sigconf]{acmart}
\usepackage{multirow}
\usepackage{float}
\usepackage{url}
\usepackage{algorithm}
\usepackage{algpseudocode}
\usepackage{bm}
\usepackage{graphicx}
\usepackage{caption}
\usepackage[skip=0cm,list=true,labelfont=it]{subcaption}
\usepackage{enumitem}

\setlist[itemize]{leftmargin=8pt,labelindent=0pt}

\usepackage{amsthm}
\newtheorem{remark}{Remark}

\theoremstyle{definition}
\newtheorem{definition}{Definition}[subsection]

\AtBeginDocument{%
  \providecommand\BibTeX{{%
    \normalfont B\kern-0.5em{\scshape i\kern-0.25em b}\kern-0.8em\TeX}}}

\copyrightyear{2026}
\acmYear{2026}
\setcopyright{cc}
\setcctype{by}
\acmConference[WWW '26]{Proceedings of the ACM Web Conference 2026}{April 13--17, 2026}{Dubai, United Arab Emirates}
\acmBooktitle{Proceedings of the ACM Web Conference 2026 (WWW '26), April 13--17, 2026, Dubai, United Arab Emirates}
\acmPrice{}
\acmDOI{10.1145/3774904.3792703}
\acmISBN{979-8-4007-2307-0/2026/04}

\begin{document}

\title[Diagnosing and Curing Fairness Pathologies in Cross-Domain Recommendation]{The Double-Edged Sword of Knowledge Transfer: Diagnosing and Curing Fairness Pathologies in Cross-Domain Recommendation}

\author{Yuhan Zhao}
\affiliation{
  \institution{Hong Kong Baptist University}
  \city{Hong Kong}
  \country{China}
}
\email{csyhzhao@comp.hkbu.edu.hk}
\authornote{Equal contributions.}

\author{Weixin Chen}
\affiliation{
  \institution{Hong Kong Baptist University}
  \city{Hong Kong}
  \country{China}
}
\email{cswxchen@comp.hkbu.edu.hk}
\authornotemark[1] 

\author{Li Chen}
\affiliation{
  \institution{Hong Kong Baptist University}
  \city{Hong Kong}
  \country{China}
}
\email{lichen@comp.hkbu.edu.hk}

\author{Weike Pan}
\affiliation{
  \institution{Shenzhen University}
  \city{Shenzhen}
  \country{China}
}
\email{panweike@szu.edu.cn}

\begin{abstract}
Cross-domain recommendation (CDR) offers an effective strategy for improving recommendation quality in a target domain by leveraging auxiliary signals from source domains. Nonetheless, emerging evidence shows that CDR can inadvertently heighten group-level unfairness. In this work, we conduct a comprehensive theoretical and empirical analysis to uncover why these fairness issues arise. Specifically, we identify two key challenges: (i) \textbf{Cross-Domain Disparity Transfer}, wherein existing group-level disparities in the source domain are systematically propagated to the target domain; and (ii) \textbf{Unfairness from Cross-Domain Information Gain}, where the benefits derived from cross-domain knowledge are unevenly allocated among distinct groups.

To address these two challenges, we propose a Cross-Domain Fairness Augmentation (CDFA) framework composed of two key components. Firstly, it mitigates cross-domain disparity transfer by adaptively integrating unlabeled data to equilibrate the informativeness of training signals across groups. Secondly, it redistributes cross-domain information gains via an information-theoretic approach to ensure equitable benefit allocation across groups. Extensive experiments on multiple datasets and baselines demonstrate that our framework significantly reduces unfairness in CDR without sacrificing overall recommendation performance, while even enhancing it. Our source code is publicly available at \url{https://github.com/weixinchen98/CDFA}.

\end{abstract}

\ccsdesc[500]{Information systems}
\ccsdesc[500]{Information systems~Recommender systems}

\keywords{Fairness, Cross-domain recommendation}
\maketitle

\section{Introduction}
Cross-domain recommendation (CDR) has emerged as a promising paradigm for enhancing recommendation quality in a target domain by leveraging auxiliary information from source domains~\cite{LCL24,XLW22,ZWC21,ZZZ23}. Despite its proven effectiveness, recent studies have highlighted notable fairness concerns, particularly in the form of exacerbated performance disparities among user groups defined by sensitive attributes~\cite{wu2025faircdr,chen2025leave,tang2024fairness, chen2025investigating}. Consequently, an increasing body of work has begun devising fairness-aware CDR strategies aimed at mitigating these group-level inequities. For instance, Tang et al.~\cite{tang2024fairness} were the first to identify that certain sensitive attributes can lead to biased outcomes in CDR settings. Likewise, VUG~\cite{chen2025leave} showed that users overlapping between domains tend to receive the bulk of advantages, while non-overlapping users gain little or even suffer deteriorated performance.

However, as a relatively nascent research area, current work addresses fairness issues without thoroughly probing the underlying mechanisms and distinctive characteristics of fairness in the CDR process. This gap leaves a fundamental question unanswered: \textit{\textbf{How do fairness issues arise in CDR?}}

To answer this question, we conduct systematic theoretical analyses and empirical studies. Our findings indicate that, under certain domains and data distributions, CDR does not necessarily exacerbate fairness issues and can even mitigate them. However, such favorable outcomes hinge on highly restricted conditions. In general, CDR presents fairness challenges distinct from those found in single-domain scenarios for two principal reasons, both stemming from its core mechanism of cross-domain information transmission:
\begin{itemize}
    \item \textbf{Cross-Domain Disparity Transfer}. CDR systems can inadvertently propagate pre-existing group disparities from the source to the target domain. For instance, if the source domain exhibits a pronounced performance gap favoring one demographic group over another, the transfer process may systematically inherit and reinforce these biases in the target domain. 
    \item \textbf{Unfairness from Cross-Domain Information Gain}. The transfer and fusion of information across domains can yield additional insights, but these benefits are not uniformly distributed across all groups. For instance, rural users may have limited exposure to delivery and ride-hailing services, resulting in less cross-domain information gain, whereas urban users benefit from richer data. Consequently, CDR can exacerbate information asymmetries, intensifying unfairness.
\end{itemize}

To address these challenges, we propose a novel framework, termed \textbf{C}ross-\textbf{D}omain \textbf{F}airness \textbf{A}ugmentation (CDFA), which comprises two key components. Firstly, to mitigate cross-domain disparity transfer, we design a data enhancement strategy. Exploiting CDR’s ability to jointly utilize source and target domain information, we select different unlabeled data (also referred to as user uninteracted items) as negative samples for different groups. Specifically, we introduce an estimator based on a fairness-related loss function to adaptively estimate the current fairness gap. Then, according to this result, we incorporate more informative unlabeled data as negative samples into the disadvantaged group’s training set. The magnitude of enhancement is proportional to the fairness gap, ensuring that information is balanced such that bias propagation is minimized. Secondly, to counteract unfairness arising from cross-domain information fusion, we introduce a cross-domain gain redistribution mechanism. We begin by quantifying the information gain from cross-domain learning based on information theory and derive an expression that is implementable via neural networks. We then redistribute these gains to achieve a more balanced improvement. Notably, our approach does not require modifications to specific knowledge-transfer components across domains, making it model-agnostic. This design choice ensures broad applicability to various CDR methods, allowing for improved fairness performance without sacrificing accuracy. Our principal contributions are outlined as follows:
\begin{itemize}
    \item To our knowledge, we are the first to provide a rigorous theoretical and empirical analysis uncovering two distinctive causes of fairness issues in CDR. 
    \item We propose a novel unlabeled data enhancement strategy that augments the information available to disadvantaged groups during cross-domain knowledge transfer.
    \item  We propose a cross-domain gain redistribution technique to address the unfairness. By analyzing and modeling cross-domain information gain from an information-theoretic perspective, we can effectively redistribute these gains to ensure equitable benefit allocation across different groups.
    \item Through extensive experimentation on multiple datasets and baselines, we demonstrate that our method effectively alleviates fairness concerns in CDR without sacrificing—and often improving—overall recommendation performance.
\end{itemize}

\section{Theoretical Analysis of Fairness Violations in CDR}
\label{sec:theory}

\newcommand{\Z}{\mathcal{Z}}
\newcommand{\X}{\mathcal{X}}
\newcommand{\Y}{\mathcal{Y}}
\newcommand{\U}{\mathcal{U}}
\newcommand{\G}{\mathcal{G}}
\newcommand{\D}{\mathcal{D}}
\newcommand{\R}{\mathbb{R}}
\newcommand{\E}{\mathbb{E}}
\newcommand{\Var}{\mathrm{Var}}
\newcommand{\Cov}{\mathrm{Cov}}
\newcommand{\1}{\mathbb{1}}

\subsection{Notation}

\begin{definition}[User-Oriented Group Fairness (UGF) Measure]
Let $Y \in \Y$ denote the recommendation output, and let $\G$ represent user groups defined by a sensitive attribute (e.g., gender, age). Fix a metric $d_\Y$ on $\Y$. The UGF measure is given by:
\begin{align}
\Gamma_{\text{UGF}}(\G, Y) := \sup_{o \in \mathcal{O}} \left| \mathbb{E}[o(Y) \mid \G=g_0] - \mathbb{E}[o(Y) \mid \G=g_1] \right|,
\end{align}
where $\mathcal{O} = \{ o: \mathcal{Y} \to \mathbb{R} \mid \mathrm{Lip}_{d_\Y}(o) \leq L_o, \, \|o\|_\infty \leq B \}$ is the set of bounded $L_o$-Lipschitz functions on $(\Y, d_\Y)$. In practice, $o$ can be a metric evaluating recommendation quality (e.g., NDCG or Hit Rate), provided it is well-defined on $\Y$ and satisfies the Lipschitz constraint with respect to $d_\Y$.
\end{definition}
Intuitively, UGF quantifies performance disparities across groups; a smaller value indicates improved system fairness. It is worth mentioning that UGF can be extended to multiple groups through summation, which makes it possible for our subsequent conclusions based on UGF to be extended to multi-group fairness.

\begin{definition}[Cross-Domain Recommendation System]
A cross-domain recommendation (CDR) system comprises:
\begin{align*}
&\text{Source domain: } \D_s = (\X_s, \mathcal{U}_s, \mathcal{I}_s, \G), \\
&\text{Target domain: } \D_t = (\X_t, \mathcal{U}_t, \mathcal{I}_t, \G), \\
&\text{Representation map: } \phi: \X_s \cup \X_t \to \Z, \\
&\text{Target predictor: } f_t: \Z \to \Y,
\end{align*}
where the subscript $s$ indicates source domain and $t$ target domain. Here, $\X_*$ denotes interactions, $\mathcal{U}_*$ denotes users, and $\mathcal{I}_*$ denotes items. We equip $\Z$ with a metric $d_\Z$ and assume $\phi$ is measurable.
\end{definition}
We focus on general CDR scenarios in which a subset of users overlap between the source and target domains, with the primary objective being to improve performance in the target domain—this is the prevailing use case in CDR~\cite{ZWC21, ZLZY23}. Since the target domain is our main focus, we adopt the convention that any variable not explicitly marked with $t$ (target) or $s$ (source) refers to the target domain for ease of exposition. For distributions on representations, we write $\nu_t := P(\phi(X_t))$, $\nu_s := P(\phi(X_s))$, and group-conditional distributions $\nu_t^g := P(\phi(X_t) \mid \G=g)$, $\nu_s^g := P(\phi(X_s) \mid \G=g)$.

\subsection{Cross-Domain Disparity Transfer}

We leverage the Kantorovich–Rubinstein duality (Wasserstein-1 on $(\Z, d_\Z)$) together with Lipschitz composition to derive an explicit upper bound on the target-domain UGF.

\begin{theorem}[Upper Bound under Domain Shift and Group Imbalance]
Assume $f_t: (\Z, d_\Z) \to (\Y, d_\Y)$ is $L_f$-Lipschitz and $o: (\Y, d_\Y) \to \R$ is $L_o$-Lipschitz and bounded by $B$. Then the target-domain group fairness satisfies
\begin{align}
\Gamma_{\mathrm{UGF}}^t \;\leq\; L_o L_f \, W_1(\nu_t^0, \nu_t^1).
\end{align}
Moreover, by the triangle inequality,
\begin{align}
W_1(\nu_t^0, \nu_t^1) \;\leq\; W_1(\nu_t^0, \nu_s^0) + W_1(\nu_s^0, \nu_s^1) + W_1(\nu_s^1, \nu_t^1),
\end{align}
and with the abbreviations
\begin{align}
\Delta_{ts} := W_1(\nu_t, \nu_s), \quad \delta_t^g := W_1(\nu_t^g, \nu_t), \quad \delta_s^g := W_1(\nu_s^g, \nu_s),
\end{align}
we have the bound
\begin{align}
W_1(\nu_t^g, \nu_s^g) \;\leq\; \delta_t^g + \Delta_{ts} + \delta_s^g, \quad g \in \{0,1\}.
\end{align}
Consequently,
\begin{align}
\Gamma_{\mathrm{UGF}}^t \;\leq\; L_o L_f \Big( W_1(\nu_s^0, \nu_s^1) + \delta_t^0 + \delta_t^1 + \delta_s^0 + \delta_s^1 + 2\Delta_{ts} \Big).
\label{eq:ugf_upper_bound}
\end{align}
\end{theorem}

\begin{proof}[Proof sketch]
Compose Lipschitz constants $L_o$ and $L_f$, then apply the Kantorovich–Rubinstein duality to obtain the core bound $\Gamma_{\mathrm{UGF}}^t \leq L_o L_f W_1(\nu_t^0, \nu_t^1)$; use the triangle inequality to decompose $W_1$ across source and target groups; introduce $\Delta_{ts}$, $\delta_t^g$, $\delta_s^g$ and bound $W_1(\nu_t^g, \nu_s^g) \leq \delta_t^g + \Delta_{ts} + \delta_s^g$; combine to derive Eq.~\eqref{eq:ugf_upper_bound}. Full proof appears in Appendix~\ref{proof1}.
\end{proof}

\begin{remark}
The bound \eqref{eq:ugf_upper_bound} separates the contributions from source-group disparity $W_1(\nu_s^0, \nu_s^1)$ and cross-domain shifts $(\delta_t^0 + \delta_t^1 + \delta_s^0 + \delta_s^1 + 2\Delta_{ts})$. This bound indicates that source domain group imbalance can inflate the target-domain fairness upper bound.
\end{remark}

\subsection{Unfairness from Cross-Domain Information Gain}


\begin{theorem}[Uniform Convergence for Group-wise Gain Gaps]
Let $\mathcal{H} := \{ h = o \circ f_t \circ \phi : o \in \mathcal{O} \}$ and assume $\|h\|_\infty \leq B$ for all $h \in \mathcal{H}$. Define the gain functional class $\mathcal{H}_{\mathrm{gain}} := \{ h_1 - h_0 : h_1, h_0 \in \mathcal{H} \}$ to capture cross-domain information gains (e.g., CDR vs. target-only baselines). For each group $g \in \{0,1\}$, let $\{Z_i^{(g)}\}_{i=1}^{n_g}$ be i.i.d. samples from $\nu_t^g$. With probability at least $1-\delta$ (over the sampling),
\begin{align}
\Big| \E_{\nu_t^g}[h] - \frac{1}{n_g} \sum_{i=1}^{n_g} h\big(Z_i^{(g)}\big) \Big| \;\leq\; 2 \, \mathfrak{R}_{n_g}(\mathcal{H}) + B \sqrt{\frac{\log(2/\delta)}{2 n_g}}, \quad \forall h \in \mathcal{H},
\end{align}
where, for $\mathcal{H}_{\mathrm{gain}}$, one may use the complexity bound $\mathfrak{R}_{n_g}(\mathcal{H}_{\mathrm{gain}}) \leq \mathfrak{R}_{n_g}(\mathcal{H}) + \mathfrak{R}_{n_g}(\mathcal{H})$ by closure under differences.
\end{theorem}

\begin{proof}[Proof sketch]
Apply classical symmetrization for bounded function classes to each group, yielding the empirical Rademacher complexity term; use Hoeffding-type concentration to bound deviations; combine the two groups via the triangle inequality to conclude. Full proof appears in Appendix~\ref{proof2}.
\end{proof}

\begin{remark}
This result connects directly to unfairness from cross-domain information gain.
and group-wise sample size and function class complexity determine the observed gain disparity.
\end{remark}

\subsection{Sufficient Condition for Fairness Preservation in CDR}

We provide a sufficient condition under which CDR does not worsen target-domain fairness relative to a baseline $\Gamma_{\mathrm{UGF}}^0$ (obtained, e.g., without cross-domain transfer).

\begin{theorem}[Fairness Preservation via Upper Bounds]
If
\begin{align}
L_o L_f \Big( W_1(\nu_s^0, \nu_s^1) + \delta_t^0 + \delta_t^1 + \delta_s^0 + \delta_s^1 + 2\Delta_{ts} \Big) \;\leq\; \Gamma_{\mathrm{UGF}}^0,
\label{eq:fairness_sufficient}
\end{align}
then $\Gamma_{\mathrm{UGF}}^t \leq \Gamma_{\mathrm{UGF}}^0$. Equivalently, controlling source-group disparity and cross-domain representation shifts (in Wasserstein-1) within the Lipschitz-scaled threshold \eqref{eq:fairness_sufficient} suffices to prevent fairness degradation.
\end{theorem}


\begin{proof}[Proof sketch]
Immediate by combining the upper bound in Eq.~\eqref{eq:ugf_upper_bound} with the threshold in Eq.~\eqref{eq:fairness_sufficient}. Full proof is in Appendix~\ref{proof1}.
\end{proof}

Eq.~\eqref{eq:fairness_sufficient} specifies a Lipschitz-scaled control term. Achieving fairness preservation requires simultaneously small source-domain group disparity and controlled representation shifts—conditions. In many real-world scenarios, source-domain fairness is seldom guaranteed or monitored, and group-conditional transfer patterns in CDR models can be inherently imbalanced. As a result, fairness violations may occur.

\subsection{Empirical Evaluation}
\begin{figure}[htbp]
    \centering
    \includegraphics[width=1\linewidth]{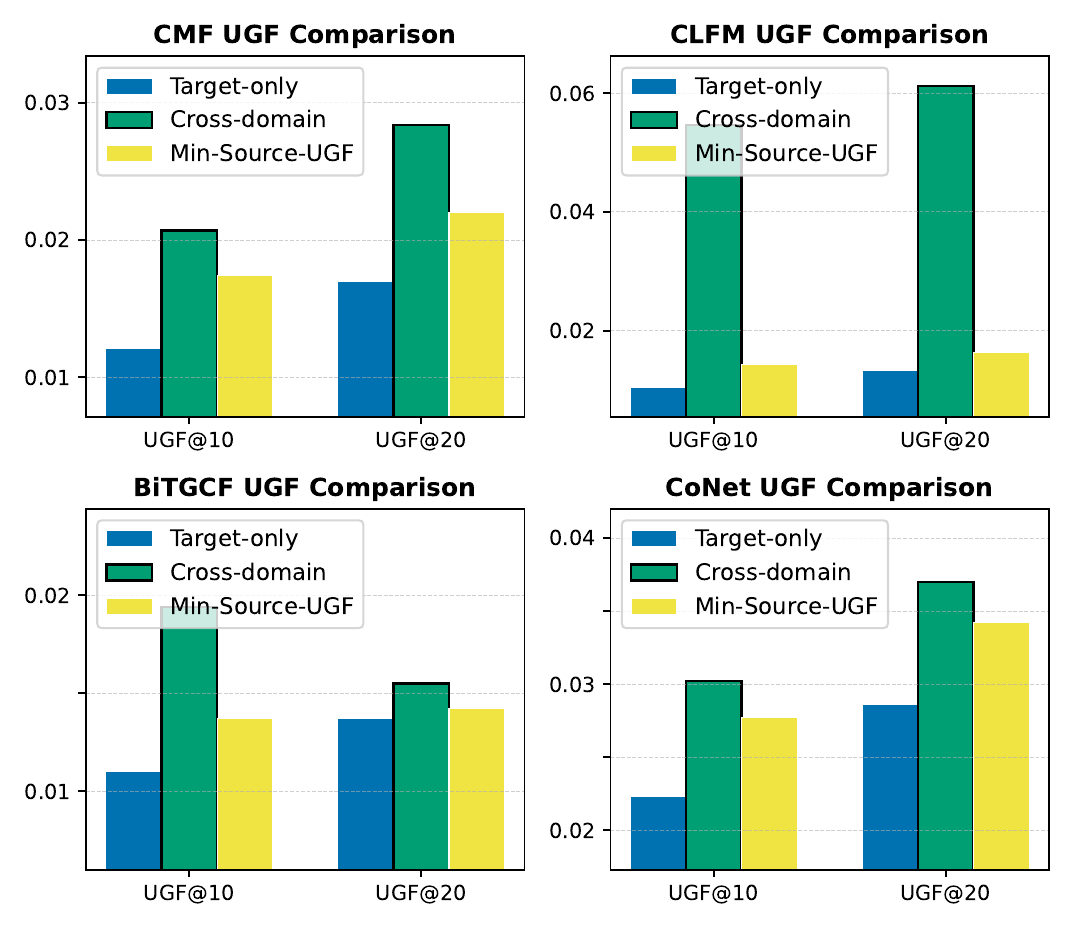}
    \caption{UGF comparison across settings on the Tenrec QB dataset, for the sensitive attribute \textit{gender}.}
    \label{fig:ugf_comparison}
    \vspace{-4mm}
\end{figure}

We empirically substantiate our theory on the Tenrec QB dataset, for which the details appear in Section~\ref{experiment}. Figure~\ref{fig:ugf_comparison} visualizes fairness performance under three settings: (i) target-only, (ii) CDR, and (iii) CDR with source-domain group disparity minimized. We observe:
\begin{itemize}
\item Comparing (i) and (ii), knowledge transfer in CDR is associated with a deterioration in target-domain fairness, suggesting that cross-domain representation shifts exacerbate fairness issues.
\item Comparing (ii) and (iii), fairness performance improves after addressing source-domain group disparity, confirming cross-domain disparity transfer. However, this alone does not fully resolve fairness challenges inherent to CDR.
\item Comparing (i) and (iii), even after mitigating cross-domain disparity transfer, the target domain still exhibits more severe fairness issues relative to (i). This indicates that new information generated during the CDR process also affects fairness, consistent with our theoretical insights.
\end{itemize}

\section{Methodology}
\begin{figure*}
    \centering
    \includegraphics[width=0.90\linewidth]{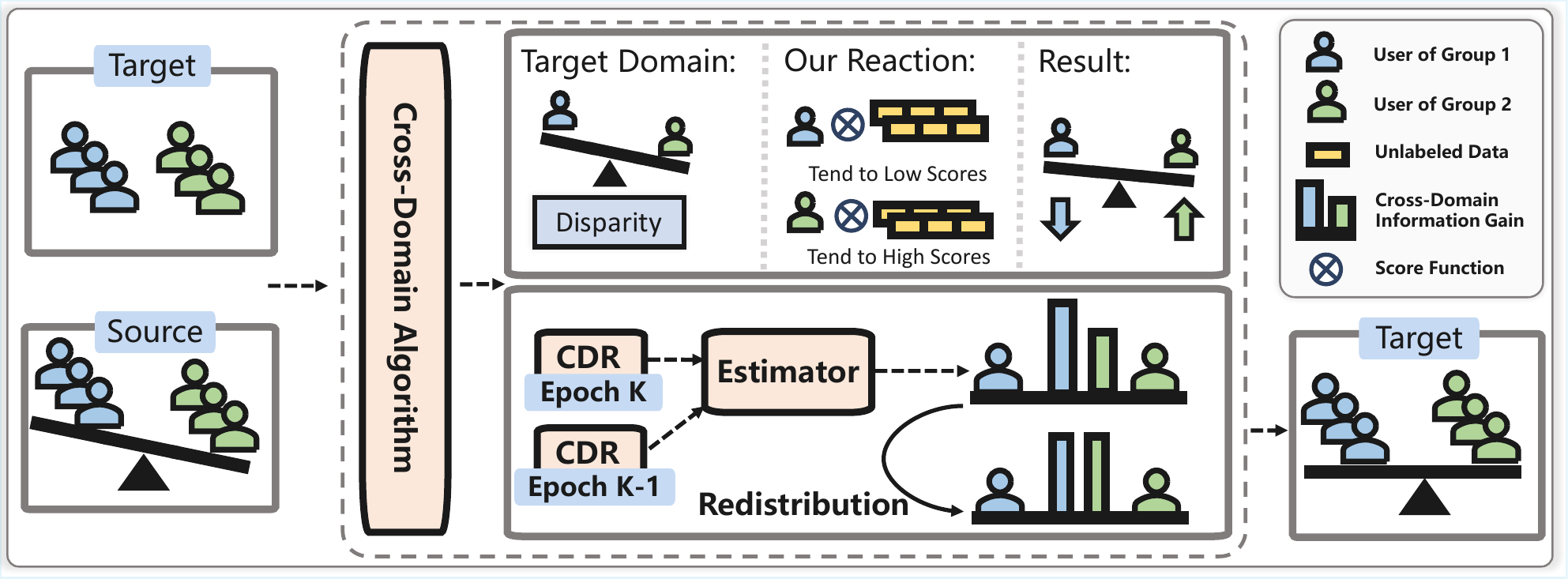}
    \vspace{-2mm}
    \caption{The illustration of our proposed Cross-Domain Fairness Augmentation (CDFA) framework.}
    \label{fig:model_illustration}
    \vspace{-3mm}
\end{figure*}
To address the aforementioned challenges, we introduce a novel \textbf{Cross-Domain Fairness Augmentation (CDFA)} framework. CDFA comprises two key modules. The first leverages unlabeled data to enhance the representation of disadvantaged groups, and the second estimates cross-domain information gain, adjusting each group to achieve parity in information acquisition. The overall workflow is illustrated in Figure~\ref{fig:model_illustration}. Notably, our approach is \textit{model-agnostic} and thus broadly applicable across diverse CDR algorithms.

\subsection{Cross-Domain Disparity Transfer}
The phenomenon of \textbf{cross-domain disparity transfer} presents a unique challenge. A straightforward, yet naïve, approach is to preprocess the source domain $\D_s$ by filtering out unfairness information~\cite{wang2022survey, chen2025causality, chen2026posttraining,  PCFR}, yielding a ``fairness-aware'' source domain:
\begin{align}
    \hat{\D}_s &= \text{filter}(\D_s).
\end{align}
Subsequently, $\hat{\D}_s$ would be used for cross-domain training. However, this approach presents two significant problems: (i) It may be suboptimal in CDR contexts, as knowledge transfer is often constrained by the limited availability of source domain information~\cite{ZLZY23,LiuZH022}. Further filtering of source information can exacerbate this limitation, potentially undermining the effectiveness of cross-domain recommendation. (ii) This approach only guarantees fairness within the source domain and neglects the impact of these changes on the target domain. To truly address fairness, we should focus on reducing disparity transfer effects within the target domain, rather than merely ensuring source domain fairness. Therefore, we refrain from adopting this approach and instead propose a novel fairness enhancement strategy that leverages unlabeled data (i.e., items with which users have not interacted). Rather than directly modifying the source domain distribution, we incorporate these unlabeled data into the CDR training process to equilibrate the informativeness of training signals across user groups.

In standard CDR training paradigms, negative sampling is employed to utilize unlabeled data: for each positive sample $i^{+}_t$, a set of unlabeled data is randomly drawn as negative samples $i^{-}_t$, and the model is trained to distinguish between them:
\begin{align}
    \min_{\Theta} \mathcal{L}(u_t, i^{+}_t, i^{-}_t).
\end{align}
$\mathcal{L}$ may be instantiated as BPR~\cite{RFG12}, cross-entropy~\cite{an2024ddcdr}, or other loss functions defined by the CDR backbone, and $\Theta$ denotes model parameters. Recent works~\cite{ZCW13,HDD21,SCF23,DQY20,ZCH24} have demonstrated that ``hard'' negative samples---i.e., items more similar to user preferences---can yield better decision boundaries. Motivated by this, we propose group-aware negative sampling: if a substantial performance gap exists between groups, we sample harder negatives for the disadvantaged group and easier negatives for the advantaged group. When group disparities are minimal, uniform sampling suffices. This mechanism mitigates performance imbalances induced by the source domain in the knowledge transfer process.

To operationalize this, we first estimate the performance gap between advantaged and disadvantaged groups. While metrics such as NDCG offer direct measurement, their non-differentiability precludes their use in end-to-end training~\cite{ZCC24,bruch2019revisiting}. Instead, inspired by~\cite{beutel2019fairness, yoo2024ensuring}, we employ the recommendation loss as a differentiable proxy for group performance:
\begin{equation}
\bar{\mathcal{L}}_g^{(K)} = 
\begin{cases}
{\mathcal{L}}_g^{(K)} & \text{if } K = 0 \\
\beta \cdot \bar{\mathcal{L}}_g^{(K-1)} + (1-\beta) \cdot {\mathcal{L}}_g^{(K)} & \text{if } K > 0
\end{cases}
\end{equation}
where $\mathcal{L}_g^{(K)}$ is the recommendation loss for group $g$ at epoch $K$. Here, we adopt a momentum-based approach rather than relying solely on the loss from the current epoch. This choice mitigates the impact of stochastic fluctuations and noise that may arise in individual epochs~\cite{ZZH22}. By leveraging momentum, we smooth the loss trajectory over recent epochs, facilitating a more reliable measure of group performance. The parameter $\beta$ governs the degree of smoothing. 
Next, for a given user $u_t$ (omitting the superscript $t$ for brevity), we compute the relative performance gap as follows:
\begin{equation}
\alpha_u
= \frac{\bar{\mathcal{L}}_{g_u}^{(K)} - \bar{\mathcal{L}}_{\text{avg}}^{(K)}}
       {\bar{\mathcal{L}}_{\text{avg}}^{(K)}},
\label{eq:alpha_def}
\end{equation}
where $\bar{\mathcal{L}}_{\text{avg}^{(K)}}$ is the average of $\bar{\mathcal{L}}_{g_1}^{(K)}$ and $\bar{\mathcal{L}}_{g_2}^{(K)}$. A larger $\alpha_u$ indicates greater group disadvantage, warranting harder negative samples. Accordingly, we modulate the negative sampling probability for user $u$ over a candidate set $M$ of unlabeled data:
\begin{align}
\label{eqn:sampling_prob}
    p(i \mid u)^{(K)} &= \frac{\exp \left(score(\mathbf{u}_t, \mathbf{i}_t) / \tau_u\right)}{\sum_{j \in M} \exp \left(score(\mathbf{u}_t, \mathbf{j}_t) / \tau_u\right)}, \quad i \in M \\
    \tau_u &= \exp(-\epsilon \alpha_u).
\end{align}

Here, $score(\cdot,\cdot)$ denotes the model's scoring function (e.g., inner product). $\mathbf{u}_*$, $\mathbf{i}_*$, and $\mathbf{j}_*$ represent the embeddings of the user, positive items, and negative items. $\epsilon$ is a hyperparameter that controls the sharpness of the distribution. As $\epsilon$ increases, the probability of sampling hard negatives correspondingly becomes higher. The exponential function $\exp(\cdot)$ ensures that the sampling probabilities remain positive. The candidate set $M$ is constructed by randomly sampling from the entire pool of available unlabeled data, with the size of $M$ specified as a hyperparameter. This strategy of constructing $M$ offers two key advantages. Firstly, directly computing negative sampling probabilities over the complete set of unlabeled data would be computationally prohibitive. Secondly, the unlabeled pool may contain items genuinely preferred by users, resulting in potential false negatives. By randomly sampling a smaller subset, we reduce the likelihood of inadvertently choosing such user-favored items as negatives, since their overall proportion within the pool remains low~\cite{ZCL23}.


\subsection{Cross-Domain Gain Redistribution}
This method addresses the varying degrees of cross-domain gains across different groups. Since groups may benefit unequally from cross-domain transfer, we seek to redistribute these gains to achieve a more balanced improvement.

To implement this, we first estimate the gains arising from cross-domain learning. From an information-theoretic perspective, the cross-domain information gain is defined as follows:
\begin{definition}[Cross-Domain Gain]
The cross-domain gain between a source and target domain is defined as follows:
\begin{align}
    \Delta I_{st} = I(Y; \X_s, \X_t) - [I(Y;\X_s) + I(Y;\X_t)],
\end{align}
where $Y$ denotes the final prediction, $I(Y; \X_s, \X_t)$ is the mutual information between $Y$ and the interactions from both domains, and $I(Y;\X_s)$ and $I(Y;\X_t)$ are the mutual information between $Y$ and each domain interaction individually. 
\end{definition}
Clearly, the difference above quantifies the additional information obtained through cross-domain learning. However, in practical scenarios, we cannot access every user's information from each domain, as not all users have behavioral records in both domains. Therefore, we estimate this value using only the overlapping users. Then we reformulate it into a form amenable to neural network estimation, as follows:
\begin{align}
\Delta I_{st} &= I(Y; \X_s,\X_t) - I(Y;\X_s) - I(Y;\X_t)  \notag \\
&= \mathbb{E}\left[\log\frac{p(y|\mathbf{x}_s,\mathbf{x}_t)}{p(y)}\right] - \mathbb{E}\left[\log\frac{p(y|\mathbf{x}_s)}{p(y)}\right] - \mathbb{E}\left[\log\frac{p(y|\mathbf{x}_t)}{p(y)}\right]  \notag \\
&= \mathbb{E}\left[ \log \frac{p(y|\mathbf{x}_s,\mathbf{x}_t)}{p(y|\mathbf{x}_s)p(y|\mathbf{x}_t)} \right] + c,
\end{align}
where $c$ is a constant that can be ignored during optimization, and $\mathbf{x}_*$ is a specific interaction in the different domain. 
It is important to note that since recommendation models are typically updated epoch by epoch, cross-domain gains may fluctuate at each epoch. We thus refine our formulation for the $K$-th epoch:
\begin{align}
    \hat{\Delta I}_{st}^{K} = \mathbb{E} \left[ \log \frac{p(y|\mathbf{x}_s^{K}, \mathbf{x}_t^{K})}{p(y|\mathbf{x}_s^{K}) p(y|\mathbf{x}_t^{K})} \right],
\end{align}
Among these terms, $p(y|\mathbf{x}_s^{K})$ and $p(y|\mathbf{x}_t^{K})$ can be readily estimated using standard approaches. For example, classic methods such as BPR~\cite{RFG12} are based on this principle. The operations are as follows:
\begin{align}
P(y|\mathbf{x}_s^{K}) &= \sigma (score(\mathbf{u}_s^K, \mathbf{i}_t^K)), \\
P(y|\mathbf{x}_t^{K}) &= \sigma (score(\mathbf{u}_t^K \cdot \mathbf{i}_t^K)),
\end{align}
where $\sigma(\cdot)$ denotes the sigmoid function,  $\mathbf{u}_*$ and $\mathbf{i}_*$ represent the embedding of user and item. To estimate $p(y|\mathbf{x}_s, \mathbf{x}_t)$, we employ a neural network to model the joint interaction:
\begin{align}
P(y|\mathbf{x}_s, \mathbf{x}_t) = \sigma (\mathrm{Estimator}(\mathbf{u}_t, \mathbf{u}_s) \cdot \mathbf{i}_t),
\end{align}
where $\mathrm{Estimator}(\cdot,\cdot)$ is a network (e.g., a multi-layer perceptron) trained to capture the new information generated by the integration of source and target representations.

A critical challenge is the lack of direct supervision for training $\mathrm{Estimator}(\cdot,\cdot)$. However, we observe that after training at epoch $K-1$, the updated target user representation $\mathbf{u}_t^K$ inherently incorporates the cross-domain gain. We thus use the previous epoch’s embeddings as input and minimize their discrepancy with the updated target embedding:
\begin{align}
\min_{\Theta_{\mathrm{Estimator}}} \| \mathrm{Estimator}(\mathbf{u}_t^{K-1}, \mathbf{u}_s^{K-1}) - \mathbf{u}_t^{K} \|^2,
\end{align}
where only the parameters of $\mathrm{Estimator}(\cdot,\cdot)$ are updated. This loss encourages the network to capture the transformation induced by cross-domain learning over one epoch.

Having estimated the cross-domain gain, we seek to balance the gains across groups. Formally, we minimize the variance in estimated gains between different groups:
\begin{align}
    \min {\|\Delta \tilde{I}_{st}^{g_1} - \Delta \tilde{I}_{st}^{g_2}\|^2},
\end{align}
where $\Delta \tilde{I}_{st}^{g_*}$ denotes the cross-domain gain of different groups. We introduce a balancing hyperparameter $\gamma$ to control the trade-off between the redistribution loss and the original loss.

\subsection{Training Procedure}
The overall training process integrates the above methods with the base cross-domain recommendation model, resulting in the following loss functions:
\begin{align}
&\min_{\Theta \setminus \Theta_{\mathrm{Estimator}}} \mathcal{L} + \gamma \|\Delta \tilde{I}_{st}^{g_1} - \Delta \tilde{I}_{st}^{g_2}\|^2, \\
&\min_{\Theta_{\mathrm{Estimator}}} \| \mathrm{Estimator}(\mathbf{u}_t^{K-1}, \mathbf{u}_s^{K-1}) - \mathbf{u}_t^{K} \|^2,
\end{align}
where $\mathcal{L}$ is the standard recommendation loss. Notably, our approach does not modify the underlying architecture of the CDR model, thus enabling broad applicability across various cross-domain recommendation frameworks.

{\textbf{Time Complexity.} Although the two main modules in our framework introduce extra computations, the overall time complexity remains on the same order as typical cross-domain recommendation frameworks. We provide a more detailed time complexity analysis in Appendix~\ref{app:timecomplexity}.}

\section{EXPERIMENTS}
\label{experiment}
In this section, we conduct comprehensive experiments to address the following research questions: \textbf{RQ1:} How does integrating CDFA into various CDR models help alleviate the unfairness? \textbf{RQ2:} How does the performance of CDFA compare with that of the state-of-the-art (SOTA) fairness-aware CDR methods? \textbf{RQ3:} How do the different components of CDFA affect the performance? \textbf{RQ4:} How do different parameter settings influence the performance of CDFA?

\subsection{Experimental Setup}

\textbf{Datasets} We conduct experiments on two datasets: (i) \textbf{QB}~\cite{yuan2022tenrec}: Collected from the QQ Browser platform, this dataset includes user interactions from both the Video and Article domains. (ii) \textbf{QK}~\cite{yuan2022tenrec}: Sourced from the QQ Kan platform, this dataset also covers both Video and Article recommendation scenarios. 


\textbf{Baselines}. 
As our approach is model-agnostic, we validate its effectiveness by integrating CDFA into the following representative CDR methods: \textbf{CMF}~\cite{SG08}, \textbf{CLFM}~\cite{GLC13}, \textbf{BiTGCF}~\cite{LLL20}, \textbf{CoNet}~\cite{HZY18}, and \textbf{DTCDR}~\cite{zhu2019dtcdr}. Furthermore, we conduct comparisons between CDFA and several representative fairness-aware CDR methods: \textbf{FairCDR}~\cite{tang2024fairness} and \textbf{VUG}~\cite{chen2025leave}.

To ensure reproducibility, our experiments are conducted based on the RecBole CDR framework~\cite{ZMH21}. Due to space constraints, detailed experimental setups—including dataset statistics, implementation, and descriptions of baselines—are provided in Appendix~\ref{app:exp}.

\begin{table*}[t]
\small
\caption{Experimental results of different CDR models with (w/) or without (w/o) our CDFA. The best results are highlighted in bold, and all improvements over the best-performing baseline(s) are statistically significant ($p < 0.05$).}
\centering
\begin{tabular}{@{}clllllllllll@{}}
\toprule
\multirow{2}{*}{\textbf{Datasets}} & \multirow{2}{*}{\textbf{Metric}} & \multicolumn{2}{c}{\textbf{CMF}} & \multicolumn{2}{c}{\textbf{CLFM}} & \multicolumn{2}{c}{\textbf{BiTGCF}} & \multicolumn{2}{c}{\textbf{CoNet}} & \multicolumn{2}{c}{\textbf{DTCDR}} \\
 \cmidrule(l){3-12} 
  &  & w/o & w/ & w/o & w/ & w/o & w/ & w/o & w/ & w/o & w/ \\ \midrule
\multirow{10}{*}{QB - gender}
 & Recall@10 & 0.2461 & \textbf{0.2715} & 0.2049 & \textbf{0.2825} & 0.2677 & \textbf{0.2754} & 0.2266 & \textbf{0.2381} & 0.2138 & \textbf{0.2324} \\
 & Recall@20 & 0.3767 & \textbf{0.4139} & 0.3070 & \textbf{0.4223} & 0.4062 & \textbf{0.4138} & 0.3596 & \textbf{0.3801} & 0.3174 & \textbf{0.3636} \\
 & NDCG@10 & 0.1241 & \textbf{0.1411} & 0.1023 & \textbf{0.1464} & 0.1398 & \textbf{0.1430} & 0.1133 & \textbf{0.1195} & 0.1084 & \textbf{0.1167} \\
 & NDCG@20 & 0.1571 & \textbf{0.1771} & 0.1281 & \textbf{0.1817} & 0.1748 & \textbf{0.1779} & 0.1468 & \textbf{0.1551} & 0.1347 & \textbf{0.1498} \\
\cmidrule(l){2-12}
 & Acc. Impr. &  & \textbf{+11.66\%} &  & \textbf{+40.09\%} &  & \textbf{+2.20\%} &  & \textbf{+5.48\%} &  & \textbf{+10.53\%}\\ \cmidrule(l){2-12}
 & UGF (Recall@10) & 0.0381 & \textbf{0.0008} & 0.1111 & \textbf{0.0015} & 0.0364 & \textbf{0.0133} & 0.0513 & \textbf{0.0005} & 0.0946 & \textbf{0.0388} \\
 & UGF (Recall@20) & 0.0689 & \textbf{0.0006} & 0.1365 & \textbf{0.0131} & 0.0219 & \textbf{0.0133} & 0.0773 & \textbf{0.0013} & 0.1312 & \textbf{0.0440} \\
 & UGF (NDCG@10) & 0.0207 & \textbf{0.0008} & 0.0546 & \textbf{0.0016} & 0.0194 & \textbf{0.0025} & 0.0302 & \textbf{0.0002} & 0.0541 & \textbf{0.0220} \\
 & UGF (NDCG@20) & 0.0284 & \textbf{0.0006} & 0.0612 & \textbf{0.0020} & 0.0155 & \textbf{0.0023} & 0.0370 & \textbf{0.0003} & 0.0636 & \textbf{0.0234} \\
\cmidrule(l){2-12}
 & Fair. Impr. &  & \textbf{+97.76\%} &  & \textbf{+95.71\%} &  & \textbf{+68.75\%} &  & \textbf{+98.97\%} &  & \textbf{+62.00\%} \\ \midrule
\multirow{10}{*}{QB - age}
 & Recall@10 & 0.2461 & \textbf{0.2511} & 0.2049 & \textbf{0.2514} & 0.2677 & \textbf{0.2736} & 0.2266 & \textbf{0.2276} & \textbf{0.2138} & 0.2112 \\
 & Recall@20 & 0.3767 & \textbf{0.3842} & 0.3070 & \textbf{0.3877} & 0.4062 & \textbf{0.4100} & 0.3596 & \textbf{0.3628} & 0.3174 & \textbf{0.3213} \\
 & NDCG@10 & 0.1241 & \textbf{0.1301} & 0.1023 & \textbf{0.1284} & 0.1398 & \textbf{0.1422} & 0.1133 & \textbf{0.1138} & \textbf{0.1084} & 0.1083 \\
 & NDCG@20 & 0.1571 & \textbf{0.1636} & 0.1281 & \textbf{0.1627} & 0.1748 & \textbf{0.1766} & 0.1468 & \textbf{0.1479} & 0.1347 & \textbf{0.1361} \\
\cmidrule(l){2-12}
 & Acc. Impr. &  & \textbf{+3.25\%} &  & \textbf{+25.38\%} &  & \textbf{+1.47\%} &  & \textbf{+0.63\%} &  & \textbf{+0.24\%}\\ \cmidrule(l){2-12}
 & UGF (Recall@10) & 0.0196 & \textbf{0.0166} & 0.0195 & \textbf{0.0061} & 0.0184 & \textbf{0.0057} & 0.0168 & \textbf{0.0007} & 0.0129 & \textbf{0.0089} \\
 & UGF (Recall@20) & 0.0245 & \textbf{0.0160} & 0.0357 & \textbf{0.0126} & 0.0315 & \textbf{0.0145} & 0.0205 & \textbf{0.0122} & 0.0255 & \textbf{0.0112} \\
 & UGF (NDCG@10) & 0.0084 & \textbf{0.0061} & 0.0078 & \textbf{0.0018} & 0.0129 & \textbf{0.0020} & 0.0078 & \textbf{0.0029} & 0.0109 & \textbf{0.0081} \\
 & UGF (NDCG@20) & 0.0096 & \textbf{0.0062} & 0.0120 & \textbf{0.0036} & 0.0164 & \textbf{0.0041} & 0.0086 & \textbf{0.0059} & 0.0140 & \textbf{0.0093} \\
\cmidrule(l){2-12}
 & Fair. Impr. &  & \textbf{+28.20\%} &  & \textbf{+70.09\%} &  & \textbf{+70.62\%} &  & \textbf{+57.63\%} &  & \textbf{+36.59\%} \\ \midrule
\multirow{10}{*}{QK - gender}
 & Recall@10 & 0.0210 & \textbf{0.0260} & 0.0132 & \textbf{0.0277} & \textbf{0.0264} & 0.0260 & \textbf{0.0129} & 0.0128 & \textbf{0.0148} & 0.0145 \\
 & Recall@20 & 0.0384 & \textbf{0.0430} & 0.0231 & \textbf{0.0430} & \textbf{0.0418} & 0.0405 & \textbf{0.0237} & 0.0224 & 0.0259 & \textbf{0.0262} \\
 & NDCG@10 & 0.0106 & \textbf{0.0136} & 0.0067 & \textbf{0.0151} & \textbf{0.0139} & 0.0138 & \textbf{0.0064} & 0.0063 & 0.0071 & \textbf{0.0073} \\
 & NDCG@20 & 0.0151 & \textbf{0.0180} & 0.0093 & \textbf{0.0191} & \textbf{0.0180} & 0.0176 & \textbf{0.0092} & 0.0088 & 0.0100 & \textbf{0.0104} \\
\cmidrule(l){2-12}
 & Acc. Impr. &  & \textbf{+20.82\%} &  & \textbf{+106.69\%} &  & \textbf{-1.89\%} &  & \textbf{-3.04\%} &  & \textbf{+1.49\%}\\ \cmidrule(l){2-12}
 & UGF (Recall@10) & 0.0014 & \textbf{0.0006} & 0.0019 & \textbf{0.0002} & 0.0016 & \textbf{0.0015} & \textbf{0.0013} & 0.0025 & 0.0048 & \textbf{0.0007} \\
 & UGF (Recall@20) & 0.0034 & \textbf{0.0016} & 0.0017 & \textbf{0.0005} & \textbf{0.0023} & 0.0037 & 0.0033 & \textbf{0.0028} & 0.0058 & \textbf{0.0003} \\
 & UGF (NDCG@10) & 0.0009 & \textbf{0.0000} & 0.0004 & \textbf{0.0004} & 0.0015 & \textbf{0.0007} & 0.0010 & \textbf{0.0005} & 0.0023 & \textbf{0.0007} \\
 & UGF (NDCG@20) & 0.0013 & \textbf{0.0002} & 0.0003 & \textbf{0.0003} & 0.0017 & \textbf{0.0001} & 0.0015 & \textbf{0.0006} & 0.0026 & \textbf{0.0005} \\
\cmidrule(l){2-12}
 & Fair. Impr. &  & \textbf{+73.67\%} &  & \textbf{+40.02\%} &  & \textbf{+23.21\%} &  & \textbf{+8.21\%} &  & \textbf{+82.64\%} \\ \midrule
\multirow{10}{*}{QK - age}
 & Recall@10 & 0.0210 & \textbf{0.0282} & 0.0132 & \textbf{0.0212} & 0.0264 & \textbf{0.0265} & 0.0129 & \textbf{0.0131} & \textbf{0.0148} & 0.0146 \\
 & Recall@20 & 0.0384 & \textbf{0.0455} & 0.0231 & \textbf{0.0372} & \textbf{0.0418} & 0.0397 & 0.0237 & \textbf{0.0238} & 0.0259 & \textbf{0.0260} \\
 & NDCG@10 & 0.0106 & \textbf{0.0147} & 0.0067 & \textbf{0.0110} & 0.0139 & \textbf{0.0146} & 0.0064 & \textbf{0.0065} & 0.0071 & \textbf{0.0071} \\
 & NDCG@20 & 0.0151 & \textbf{0.0192} & 0.0093 & \textbf{0.0152} & 0.0180 & \textbf{0.0181} & 0.0092 & \textbf{0.0092} & 0.0100 & \textbf{0.0101} \\
\cmidrule(l){2-12}
 & Acc. Impr. &  & \textbf{+29.65\%} &  & \textbf{+62.32\%} &  & \textbf{+0.24\%} &  & \textbf{+0.88\%} &  & \textbf{+0.01\%}\\ \cmidrule(l){2-12}
 & UGF (Recall@10) & 0.0018 & \textbf{0.0007} & 0.0015 & \textbf{0.0001} & 0.0017 & \textbf{0.0004} & 0.0025 & \textbf{0.0018} & 0.0030 & \textbf{0.0010} \\
 & UGF (Recall@20) & 0.0003 & \textbf{0.0003} & 0.0059 & \textbf{0.0001} & \textbf{0.0023} & 0.0024 & 0.0023 & \textbf{0.0021} & 0.0045 & \textbf{0.0013} \\
 & UGF (NDCG@10) & 0.0014 & \textbf{0.0005} & 0.0010 & \textbf{0.0001} & 0.0014 & \textbf{0.0006} & 0.0008 & \textbf{0.0006} & 0.0014 & \textbf{0.0001} \\
 & UGF (NDCG@20) & 0.0010 & \textbf{0.0003} & 0.0022 & \textbf{0.0000} & 0.0016 & \textbf{0.0001} & 0.0008 & \textbf{0.0007} & 0.0018 & \textbf{0.0002} \\
\cmidrule(l){2-12}
 & Fair. Impr. &  & \textbf{+48.85\%} &  & \textbf{+95.41\%} &  & \textbf{+55.75\%} &  & \textbf{+18.55\%} &  & \textbf{+79.88\%} \\ \bottomrule
\end{tabular}
\label{tab:main_results}
\end{table*}

\subsection{Overall Performance Comparison}
We report the overall performance results in Table~\ref{tab:main_results}. All reported improvements represent average gains across all evaluation metrics, and all improvements are statistically significant ($p < 0.05$) compared to the best-performing baseline(s), as assessed by a two-sided t-test. Key observations include:
\begin{itemize}
    \item \textbf{Significant Improvement in Fairness:} With the integration of CDFA, all base models exhibit notable improvements in UGF-related metrics. For example, CMF achieves a performance increase of up to 98.97\% on QB. These results demonstrate CDFA’s effectiveness in mitigating group unfairness by addressing the two key challenges we identified.

    \item \textbf{Preserved or Improved Accuracy:} Notably, while CDFA substantially enhances fairness, it also improves accuracy in most scenarios. In rare cases where a decrease occurs, the maximum decline is just 3.04\%, which is a negligible compromise. We attribute this to our method’s two-part design: the augmentation of information via unlabeled data and the redistribution of information gains. This increases the total information utilized without diminishing any CDR information, thus achieving a balance between fairness and accuracy. Such a property makes CDFA suitable and attractive for real-world industrial applications.

    \item \textbf{SOTA Methods and Fairness Trade-off:} Interestingly, even state-of-the-art methods underperform on fairness-related metrics, sometimes even lagging behind traditional approaches such as BiTGCF and CLFM. This finding underscores the risk of disregarding fairness during model design: improvements in accuracy may come at the expense of fairness, thereby potentially harming user experience and satisfaction.
\end{itemize}

\subsection{Performance Comparison with Other Fairness-aware CDR Methods}  
As illustrated on Figure~\ref{fig:fairCDR}, CDFA consistently outperforms other fairness-aware CDR methods in both accuracy and fairness metrics. Importantly, CDFA does not incur any loss in accuracy. This advantage can be attributed to two main factors: (i) Unlike previous approaches that focus solely on mitigating unfairness, CDFA is designed with a thorough understanding of how fairness is generated and amplified within the CDR context. This principled approach enables CDFA to more effectively address fairness challenges inherent to CDR. (ii) CDFA has no lost information. The first module leverages CDR to mine informative signals from unlabeled data, thereby introducing additional information. This design facilitates a strong balance between accuracy and fairness.

\begin{figure*}[htbp]
    \centering
    \includegraphics[width=1\linewidth]{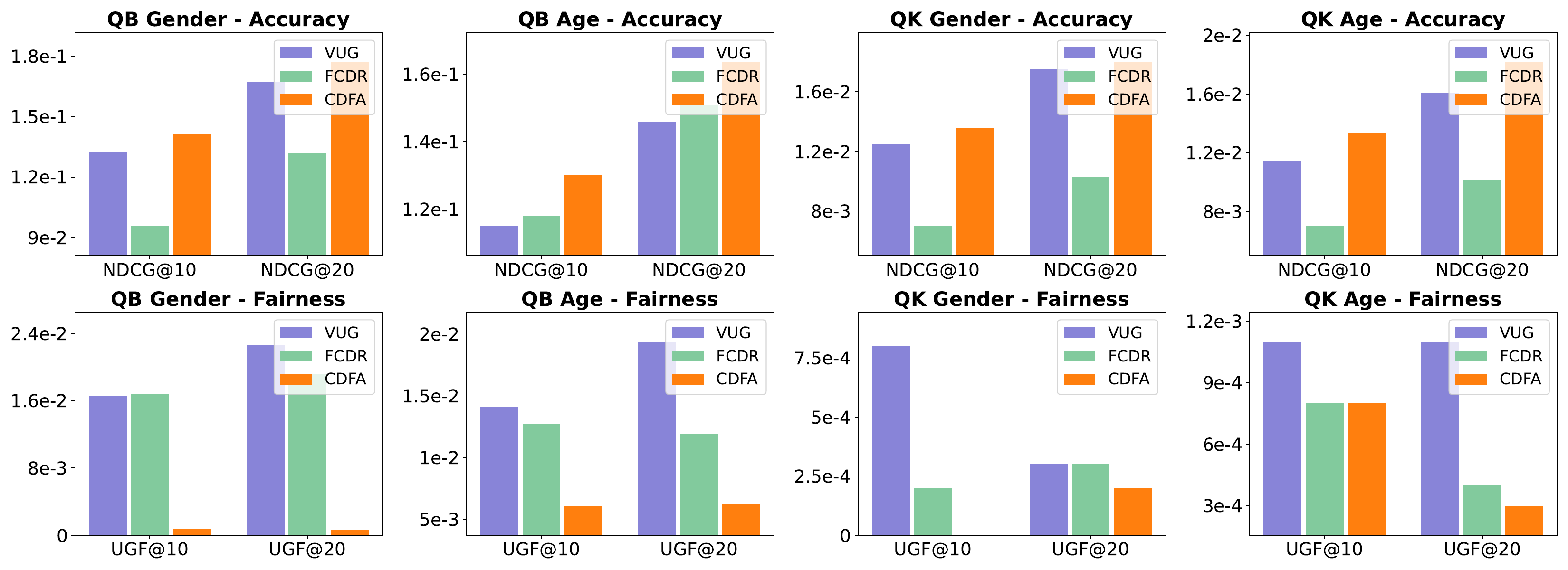}
    \vspace{-5mm}
    \caption{Comparison of fairness-aware cross-domain recommendation baselines on the Tenrec QB and QK datasets.}
    \label{fig:fairCDR}
    \vspace{-4mm}
\end{figure*}

\subsection{Ablation Study}
\label{sec:component_analysis}
We analyze the effectiveness of different components via the following variants: (i) CDFA without $\alpha_u$ (w/o $\alpha_u$), (ii) CDFA without sampling (w/o FS, using random sampling), (iii) CDFA without redistribution loss (w/o $\mathcal{L}_{redist}$), where $\mathcal{L}_{redist} = \|\Delta \tilde{I}_{st}^{G_1} - \Delta \tilde{I}_{st}^{G_2}\|^2$, and (iv) CDFA without estimator loss (w/o $\mathcal{L}_{est}$),  where $\mathcal{L}_{est} = \| \mathrm{Estimator}(\mathbf{u}_t^{K-1}, \mathbf{u}_s^{K-1}) - \mathbf{u}_t^{K} \|^2$. The results are presented in Table~\ref{tab:ablation_gender}. They indicate that all components make positive contributions to model fairness, which is expected given that each module in CDFA is explicitly designed to enhance fairness. In terms of accuracy, we observe some performance fluctuations when CDFA is applied without $\mathcal{L}{est}$. This can be explained by the fact that, during this phase, estimated cross-domain information gain must be redistributed. Since estimation may not be perfectly precise, the redistribution may be suboptimal, resulting in occasional deviations in accuracy. 

\begin{table}[t]
\small
\centering
\caption{Ablation study of CDFA, integrated with CMF on the Tenrec QB dataset for gender fairness.}
\vspace{-2mm}
\label{tab:ablation_gender}
\begin{tabular}{lcccc}
\toprule
\multirow{2}{*}{\textbf{Method}} & \multicolumn{4}{c}{\textbf{Accuracy (\textit{larger} is better)}}\\
\cmidrule(lr){2-5}
& Recall@10 & Recall@20 & NDCG@10 & NDCG@20 \\ 
\midrule
\emph{w/o $\alpha_u$} & 0.2727 & 0.4137 & 0.1414 & 0.1770 \\ 
\emph{w/o FS} & 0.2571 & 0.3950 & 0.1322 & 0.1670 \\ 
\emph{w/o $\mathcal{L}_{redist}$} & 0.2526 & 0.3905 & 0.1308 & 0.1656 \\ 
\emph{w/o $\mathcal{L}_{est}$} & \textbf{0.2729} & 0.4128 & \textbf{0.1414} & 0.1768 \\ 
\textbf{CDFA} & {0.2715} & \textbf{0.4139} & {0.1411} & \textbf{0.1771} \\ 
\midrule
\multirow{2}{*}{\textbf{Method}} & \multicolumn{4}{c}{\textbf{UGF (\textit{smaller} is better)}}\\
\cmidrule(lr){2-5}
& Recall@10 & Recall@20 & NDCG@10 & NDCG@20 \\
\midrule
\emph{w/o $\alpha_u$} & 0.0044 & 0.0020 & 0.0011 & 0.0007 \\
\emph{w/o FS} & 0.0393 & 0.0458 & 0.0237 & 0.0255 \\
\emph{w/o $\mathcal{L}_{redist}$} & 0.0471 & 0.0541 & 0.0295 &  0.0315 \\
\emph{w/o $\mathcal{L}_{est}$} & 0.0035 & 0.0024 & 0.0009 & 0.0007 \\ 
\textbf{CDFA} & \textbf{0.0008} & \textbf{0.0006} & \textbf{0.0008} & \textbf{0.0006} \\
\bottomrule
\end{tabular}
\vspace{-4mm}
\end{table}

\subsection{Hyperparameter Sensitivity Analysis}
\label{sec:parameter_sensitivity}
The three most critical hyperparameters in CDFA are: the size of $M$, which denotes the candidate set size during negative sampling (i.e., a larger $M$ means more unlabeled data considered); $\epsilon$, which regulates the strength of negative sampling, with higher values favoring harder negative instances; and $\gamma$, which determines the extent of gain redistribution, with a larger value yielding a more equitable distribution of gains across groups. The experimental results are shown in Figure~\ref{fig:param_sensitivity}. Three key observations emerge: (i) Within reasonable ranges, all hyperparameters allow CDFA to maintain robust accuracy and fairness, indicating relative stability. (ii) With $M$, fairness performance remains consistent, while accuracy exhibits noticeable fluctuations. This is understandable, as increasing $M$ risks introducing more false negatives into the candidate set, potentially distorting user preference learning. Although fairness is largely preserved, the accuracy may suffer due to incorrect user modeling. (iii) The influence of $\gamma$ on performance is more pronounced, with accuracy and fairness displaying a degree of mutual exclusivity. This aligns with intuition, as $\gamma$ mediates the trade-off between fairness-related and accuracy-related loss components. 
When $\gamma$ is high, the model prioritizes fairness optimization and may neglect accuracy, while lower values may result in diminished fairness improvements.

\begin{figure}[htbp]
    \centering
    \includegraphics[width=1\linewidth]{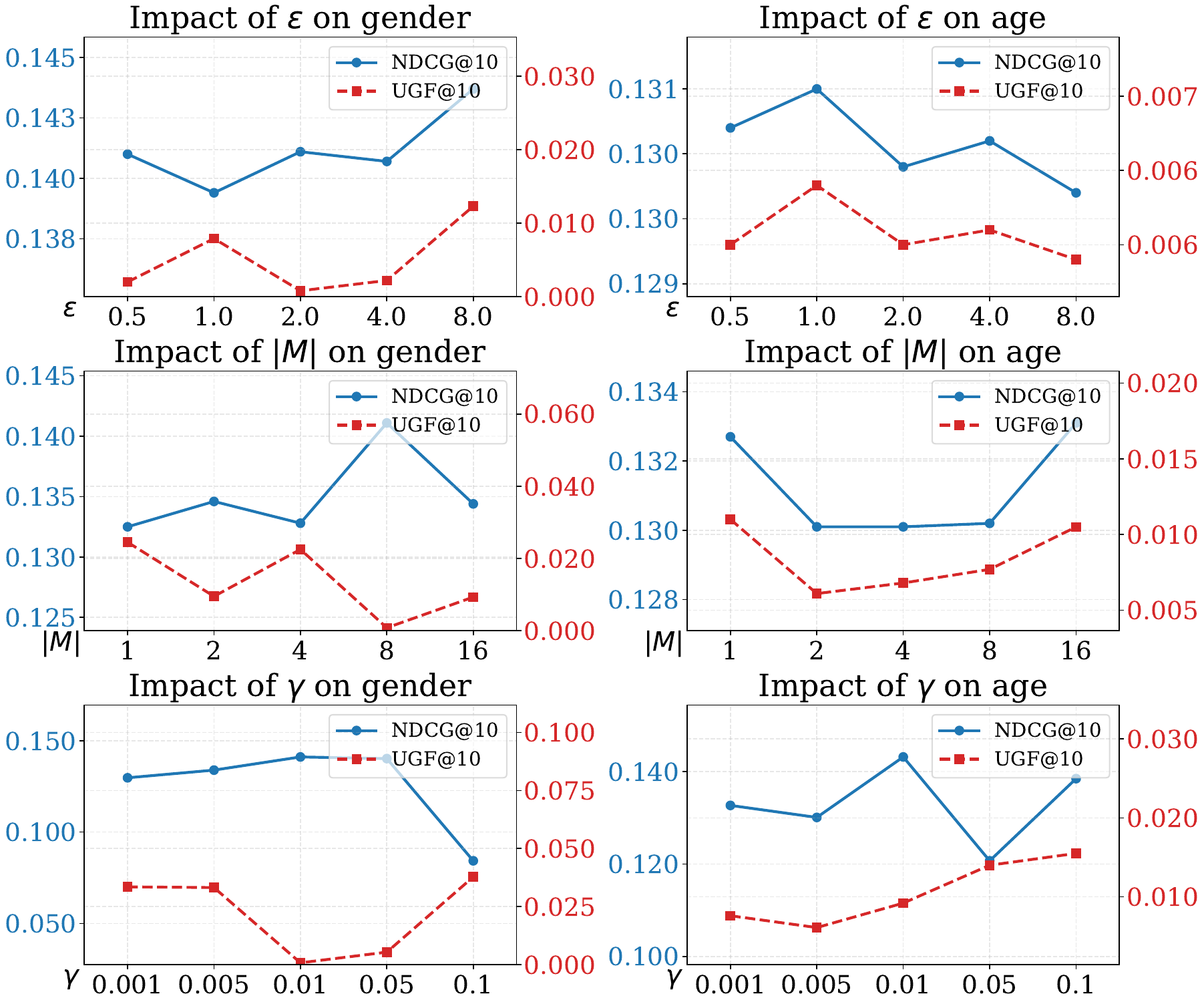}
    \vspace{-3mm}
    \caption{Impact of the hyperparameters on our CDFA on the Tenrec QB dataset.}
    \label{fig:param_sensitivity}
    \vspace{-4mm}
\end{figure}

\section{Related Work}

\subsection{Cross-Domain Recommendation}
Cross-domain recommendation (CDR) aims to leverage information from a source domain to improve the accuracy of recommendations in a target domain. For example, CMF~\cite{SG08} was among the first to jointly factorize multiple rating matrices by sharing latent parameters across domains.
PPCDR~\cite{tian2024privacy} adopts a federated graph learning approach to capture user preferences to ensure privacy. GDCCDR~\cite{liu2024graph} achieves personalized transfer through meta-networks and contrastive learning. PCDR~\cite{guo2024prompt} utilizes users' interactions on local clients, employs gradient encryption, and harnesses natural language to establish bridges between domains. MACDR~\cite{WYW24} explores unsupervised utilization of non-overlapping users. MFGSLAE~\cite{WYZ25} proposes a factor selection module with a bootstrapping mechanism to identify domain-specific preferences.

\subsection{Fairness in Cross-Domain Recommendation}
Although CDR has demonstrated substantial effectiveness, recent studies have brought to light notable fairness concerns—particularly the exacerbation of performance disparities among user groups defined by sensitive attributes. FairCDR~\cite{tang2024fairness} identifies unfair treatment of cold-start users in the target domain and proposes a novel algorithm that models the individual influence of each sample on both fairness and accuracy, thereby facilitating data reweighting. VUG~\cite{chen2025leave} reveals a critical bias in CDR: while overlapping users benefit substantially from improved recommendation quality, non-overlapping users gain minimally and may even experience performance degradation. To address this issue, they introduced a solution that generates virtual source-domain users for non-overlapping target-domain users to enhance the CDR process. FairnessCDR~\cite{wu2025faircdr} ensures fairness using adversarial learning and a mutual learning-based diffusion model.

\section{Conclusion}
In this work, through comprehensive theoretical analysis and empirical evaluation, we identify two fundamental mechanisms underlying fairness concerns in cross-domain recommendation (CDR): (i) cross-domain disparity transfer and (ii) unfairness arising from cross-domain information gain. To address these challenges, we propose the Cross-Domain Fairness Augmentation (CDFA) framework. CDFA first utilizes unlabeled data to balance the informativeness of training signals across user groups. It then adopts an information-theoretic approach to redistribute cross-domain information gains, ensuring equitable benefit across diverse groups. Extensive experiments on multiple datasets and strong baselines demonstrate that CDFA substantially reduces unfairness in CDR, while maintaining or even improving overall recommendation performance. Our approach marks an important step toward making CDR systems fairer and more trustworthy for users.
In the future, we also plan to investigate fairness issues in emerging agent-based recommendation paradigms~\cite{chen2026memrec, xu2025iagent}.

\begin{acks}
This work is supported by Hong Kong Baptist University Key Research Partnership Scheme (KRPS/23-24/02), NSFC/RGC Joint Research Scheme (N\_HKBU214/24), and National Natural Science Foundation of China (62461160311).
\end{acks}

\clearpage

\balance
\bibliographystyle{ACM-Reference-Format}
\bibliography{paper}

@inproceedings{guo2024prompt,
  title={Prompt-Enhanced Federated Content Representation Learning for Cross-Domain Recommendation},
  author={Guo, Lei and Lu, Ziang and Yu, Junliang and Nguyen, Quoc Viet Hung and Yin, Hongzhi},
  booktitle={WWW},
  pages={3139--3149},
  year={2024}
}

@inproceedings{liu2024graph,
  title={Graph Disentangled Contrastive Learning with Personalized Transfer for Cross-Domain Recommendation},
  author={Liu, Jing and Sun, Lele and Nie, Weizhi and Jing, Peiguang and Su, Yuting},
  booktitle={AAAI},
  volume={38},
  number={8},
  pages={8769--8777},
  year={2024}
}

@article{tian2024privacy,
  title={Privacy-Preserving Cross-domain Recommendation with Federated Graph Learning},
  author={Tian, Changxin and Xie, Yuexiang and Chen, Xu and Li, Yaliang and Zhao, Xin},
  journal={TOIS},
  volume={42},
  number={5},
  pages={1--29},
  year={2024},
  publisher={ACM New York, NY}
}

@inproceedings{bruch2019revisiting,
  title={Revisiting Approximate Metric Optimization in the Age of Deep Neural Networks},
  author={Bruch, Sebastian and Zoghi, Masrour and Bendersky, Michael and Najork, Marc},
  booktitle={SIGIR},
  pages={1241--1244},
  year={2019}
}

@inproceedings{an2024ddcdr,
  title={DDCDR: A Disentangle-based Distillation Framework for Cross-Domain Recommendation},
  author={An, Zhicheng and Gu, Zhexu and Yu, Li and Tu, Ke and Wu, Zhengwei and Hu, Binbin and Zhang, Zhiqiang and Gu, Lihong and Gu, Jinjie},
  booktitle={KDD},
  pages={4764--4773},
  year={2024}
}

@inproceedings{chen2025leave,
  title={Leave No One Behind: Fairness-Aware Cross-Domain Recommender Systems for Non-Overlapping Users},
  author={Chen, Weixin and Zhao, Yuhan and Chen, Li and Pan, Weike},
  booktitle={RecSys},
  pages={226--236},
  year={2025}
}

@inproceedings{wu2025faircdr,
  title={FairCDR: Transferring Fairness and User Preferences for Cross-Domain Recommendation},
  author={Wu, Yongxuan and Liu, Yang and Lin, Xixun and Zhou, Hong and Cao, Yanan and Zou, Lixin and Shang, Yanmin and Liu, Yanbing},
  booktitle={KDD},
  pages={3261--3272},
  year={2025}
}

@inproceedings{GLC13,
  title={Cross-Domain Recommendation via Cluster-Level Latent Factor Model},
  author={Gao, Sheng and Luo, Hao and Chen, Da and Li, Shantao and Gallinari, Patrick and Guo, Jun},
  booktitle={PKDD},
  pages={161--176},
  year={2013}
}

@inproceedings{ZZZ23,
  title={Disentangled Contrastive Learning for Cross-Domain Recommendation},
  author={Zhang, Ruohan and Zang, Tianzi and Zhu, Yanmin and Wang, Chunyang and Wang, Ke and Yu, Jiadi},
  booktitle={DASFAA},
  pages={163--178},
  year={2023},
  organization={Springer}
}

@inproceedings{XLW22,
  title={Contrastive Cross-Domain Recommendation in Matching},
  author={Xie, Ruobing and Liu, Qi and Wang, Liangdong and Liu, Shukai and Zhang, Bo and Lin, Leyu},
  booktitle={KDD},
  pages={4226--4236},
  year={2022}
}

@article{ZWC21,
  title={Cross-Domain Recommendation: Challenges, Progress, and Prospects},
  author={Zhu, Feng and Wang, Yan and Chen, Chaochao and Zhou, Jun and Li, Longfei and Liu, Guanfeng},
  journal={arXiv preprint arXiv:2103.01696},
  year={2021}
}

@inproceedings{LCL24,
  title={User Distribution Mapping Modelling with Collaborative Filtering for Cross Domain Recommendation},
  author={Liu, Weiming and Chen, Chaochao and Liao, Xinting and Hu, Mengling and Su, Jiajie and Tan, Yanchao and Wang, Fan},
  booktitle={WWW},
  pages={334--343},
  year={2024}
}

@inproceedings{SG08,
  title={Relational Learning via Collective Matrix Factorization},
  author={Singh, Ajit P and Gordon, Geoffrey J},
  booktitle={KDD},
  pages={650--658},
  year={2008}
}

@inproceedings{HZY18,
  title={CoNet: Collaborative Cross Networks for Cross-Domain Recommendation},
  author={Hu, Guangneng and Zhang, Yu and Yang, Qiang},
  booktitle={CIKM},
  pages={667--676},
  year={2018}
}

@article{WYW24,
  title={Making Non-Overlapping Matters: An Unsupervised Alignment Enhanced Cross-Domain Cold-Start Recommendation},
  author={Wang, Zihan and Yang, Yonghui and Wu, Le and Hong, Richang and Wang, Meng},
  journal={TKDE},
pages={2001--2014},
volume={37},
  year={2025},
}

@article{WYZ25,
author = {Wang, Hao and Yin, Mingjia and Zhang, Luankang and Zhao, Sirui and Chen, Enhong},
title = {MF-GSLAE: A Multi-Factor User Representation Pre-Training Framework for Dual-Target Cross-Domain Recommendation},
year = {2025},
volume = {43},
number = {2},
journal = {TOIS}
}

@inproceedings{LLL20,
  title={Cross Domain Recommendation via Bi-Directional Transfer Graph Collaborative Filtering Networks},
  author={Liu, Meng and Li, Jianjun and Li, Guohui and Pan, Peng},
  booktitle={CIKM},
  pages={885--894},
  year={2020}
}

@inproceedings{KB14,
  author       = {Diederik P. Kingma and
                  Jimmy Ba},
  title        = {Adam: {A} Method for Stochastic Optimization},
  booktitle    = {ICLR},
  year         = {2015},

}

@inproceedings{ZHP22,
  title={RecBole 2.0: Towards a More Up-to-Date Recommendation Library},
  author={Zhao, Wayne Xin and Hou, Yupeng and Pan, Xingyu and Yang, Chen and Zhang, Zeyu and Lin, Zihan and Zhang, Jingsen and Bian, Shuqing and Tang, Jiakai and Sun, Wenqi and others},
  booktitle={CIKM},
  pages={4722--4726},
  year={2022}
}

@inproceedings{ZCC24,
author = {Zhao, Yuhan and Chen, Rui and Chen, Li and Zhang, Shuang and Han, Qilong and Song, Hongtao},
title = {From Pairwise to Ranking: Climbing the Ladder to Ideal Collaborative Filtering with Pseudo-Ranking},
year = {2025},
booktitle = {AAAI},
articleno = {1489},
numpages = {9}
}

@inproceedings{SCF23,
  title={On the Theories behind Hard Negative Sampling for Recommendation},
  author={Shi, Wentao and Chen, Jiawei and Feng, Fuli and Zhang, Jizhi and Wu, Junkang and Gao, Chongming and He, Xiangnan},
  booktitle={WWW},
  pages={812--822},
  year={2023}
}

@inproceedings{ZCH24,
author = {Zhao, Yuhan and Chen, Rui and Han, Qilong and Song, Hongtao and Chen, Li},
title = {Unlocking the Hidden Treasures: Enhancing Recommendations with Unlabeled Data},
year = {2024},
booktitle = {RecSys},
pages = {247–256},
numpages = {10}
}

@inproceedings{ZCL23,
author = {Zhao, Yuhan and Chen, Rui and Lai, Riwei and Han, Qilong and Song, Hongtao and Chen, Li},
title = {Augmented Negative Sampling for Collaborative Filtering},
booktitle = {RecSys},
pages = {256–266},
year = {2023}
}

@article{DQY20,
  title={Simplify and Robustify Negative Sampling for Implicit Collaborative Filtering},
  author={Ding, Jingtao and Quan, Yuhan and Yao, Quanming and Li, Yong and Jin, Depeng},
  journal={NeurIPS},
  volume={33},
  year={2020}
}

@inproceedings{RFG12,
  author       = {Rendle, Steffen and Freudenthaler, Christoph and Gantner, Zeno and Schmidt-Thieme},
  title        = {{BPR:} Bayesian Personalized Ranking from Implicit Feedback},
  booktitle    = {UAI},
  pages        = {452--461},
  year         = {2009}

}

@inproceedings{ZZH22,
  title={A Gain-Tuning Dynamic Negative Sampler for Recommendation},
  author={Zhu, Qiannan and Zhang, Haobo and He, Qing and Dou, Zhicheng},
  booktitle={WWW},
  pages={277--285},
  year={2022}
}

@inproceedings{ZCW13,
  title={Optimizing Top-N Collaborative Filtering via Dynamic Negative Item Sampling},
  author={Zhang, Weinan and Chen, Tianqi and Wang, Jun and Yu, Yong},
  booktitle={SIGIR},
  pages={785--788},
  year={2013}
}

@inproceedings{HDD21,
  title={MixGCF: An Improved Training Method for Graph Neural Network-based Recommender Systems},
  author={Huang, Tinglin and Dong, Yuxiao and Ding, Ming and Yang, Zhen and Feng, Wenzheng and Wang, Xinyu and Tang, Jie},
  booktitle={KDD},
  pages={665--674},
  year={2021}
}

@inproceedings{ZMH21,
  title={Recbole: Towards a Unified, Comprehensive and Efficient Framework for Recommendation Algorithms},
  author={Zhao, Wayne Xin and Mu, Shanlei and Hou, Yupeng and Lin, Zihan and Chen, Yushuo and Pan, Xingyu and Li, Kaiyuan and Lu, Yujie and Wang, Hui and Tian, Changxin and others},
  booktitle={CIKM},
  pages={4653--4664},
  year={2021}
}

@inproceedings{yoo2024ensuring,
	title        = {Ensuring User-side Fairness in Dynamic Recommender Systems},
	author       = {Yoo, Hyunsik and Zeng, Zhichen and Kang, Jian and Qiu, Ruizhong and Zhou, David and Liu, Zhining and Wang, Fei and Xu, Charlie and Chan, Eunice and Tong, Hanghang},
	year         = 2024,
	booktitle    = {WWW},
	pages        = {3667--3678}
}

@inproceedings{UGF,
	title        = {User-Oriented Fairness in Recommendation},
	author       = {Yunqi Li and Hanxiong Chen and Zuohui Fu and Yingqiang Ge and Yongfeng Zhang},
	year         = 2021,
	booktitle    = {WWW},
    pages = {624--632},
}

@inproceedings{PCFR,
	title        = {Towards Personalized Fairness based on Causal Notion},
	author       = {Yunqi Li and Hanxiong Chen and Shuyuan Xu and Yingqiang Ge and Yongfeng Zhang},
	year         = 2021,
	booktitle    = {SIGIR},
    pages = {1054--1063},
}

@article{wang2022survey,
	title        = {A Survey on the Fairness of Recommender Systems},
	author       = {Wang, Yifan and Ma, Weizhi and Zhang, Min and Liu, Yiqun and Ma, Shaoping},
    journal={TOIS},
    year={2023},
      volume       = {41},
      number       = {3},
  pages        = {52:1--52:43},
}

@article{chen2025causality,
  title     = {Causality-Inspired Fair Representation Learning for Multimodal Recommendation},
  author    = {Chen, Weixin and Chen, Li and Ni, Yongxin and Zhao, Yuhan},
  year      = 2025,
  journal   = {TOIS},
  volume    = {43},
  number    = {6},
  articleno = {153},
  numpages  = {29}
}

@article{chen2025investigating,
  title     = {Investigating User-side fairness in outcome and process for multi-type sensitive attributes in recommendations},
  author    = {Chen, Weixin and Chen, Li and Zhao, Yuhan},
  year      = 2025,
  journal   = {TORS},
  volume    = {4},
  number    = {2},
  articleno = {25},
  numpages  = {29}
}

@inproceedings{tang2024fairness,
  title={Fairness-aware Cross-Domain Recommendation},
  author={Tang, Jiakai and Feng, Xueyang and Chen, Xu},
  booktitle={DASFAA},
  pages={293--302},
  year={2024},
}

@article{ZLZY23,
  author       = {Tianzi Zang and
                  Yanmin Zhu and
                  Haobing Liu and
                  Ruohan Zhang and
                  Jiadi Yu},
  title        = {A Survey on Cross-domain Recommendation: Taxonomies, Methods, and Future Directions},
  journal      = {TOIS},
  volume       = {41},
  number       = {2},
  pages        = {42:1--42:39},
  year         = {2023},
}

@inproceedings{LiuZH022,
  author       = {Weiming Liu and
                  Xiaolin Zheng and
                  Mengling Hu and
                  Chaochao Chen},
  title        = {Collaborative Filtering with Attribution Alignment for Review-based Non-overlapped Cross Domain Recommendation},
  booktitle    = {WWW},
  pages        = {1181--1190},
  year         = {2022},
}

@inproceedings{beutel2019fairness,
  title={Fairness in Recommendation Ranking through Pairwise Comparisons},
  author={Beutel, Alex and Chen, Jilin and Doshi, Tulsee and Qian, Hai and Wei, Li and Wu, Yi and Heldt, Lukasz and Zhao, Zhe and Hong, Lichan and Chi, Ed H and others},
  booktitle={KDD},
  pages={2212--2220},
  year={2019}
}

@inproceedings{zhu2019dtcdr,
  title={Dtcdr: A Framework for Dual-Target Cross-Domain Recommendation},
  author={Zhu, Feng and Chen, Chaochao and Wang, Yan and Liu, Guanfeng and Zheng, Xiaolin},
  booktitle={CIKM},
  pages={1533--1542},
  year={2019}
}

@article{yuan2022tenrec,
  title={Tenrec: A Large-Scale Multipurpose Benchmark Dataset for Recommender Systems},
  author={Yuan, Guanghu and Yuan, Fajie and Li, Yudong and Kong, Beibei and Li, Shujie and Chen, Lei and Yang, Min and Yu, Chenyun and Hu, Bo and Li, Zang and others},
  journal={NeurIPS},
  volume={35},
  pages={11480--11493},
  year={2022}
}

@inproceedings{sutskever2013importance,
  title={On the Importance of Initialization and Momentum in Deep Learning},
  author={Sutskever, Ilya and Martens, James and Dahl, George and Hinton, Geoffrey},
  booktitle={ICML},
  pages={1139--1147},
  year={2013},
}

@inproceedings{he2016deep,
  title={Deep Residual learning for Image Recognition},
  author={He, Kaiming and Zhang, Xiangyu and Ren, Shaoqing and Sun, Jian},
  booktitle={CVPR},
  pages={770--778},
  year={2016}
}

@inproceedings{xu2025heterogeneous,
  title={Heterogeneous Graph Transfer Learning for Category-aware Cross-Domain Sequential Recommendation},
  author={Xu, Zitao and Chen, Xiaoqing and Pan, Weike and Ming, Zhong},
  booktitle={WWW},
  pages={1951--1962},
  year={2025}
}

@inproceedings{lin2025towards,
  title={Towards Interest Drift-driven User Representation Learning in Sequential Recommendation},
  author={Lin, Xiaolin and Pan, Weike and Ming, Zhong},
  booktitle={SIGIR},
  pages={1541--1551},
  year={2025}
}

@inproceedings{chen2026posttraining,
  title     = {Post-Training Fairness Control: A Single-Train Framework for Dynamic Fairness in Recommendation},
  author    = {Chen, Weixin and Chen, Li and Zhao, Yuhan},
  year      = 2026,
  booktitle = {WWW Companion}
}

@article{chen2026memrec,
  title   = {MemRec: Collaborative Memory-Augmented Agentic Recommender System},
  author  = {Chen, Weixin and Zhao, Yuhan and Huang, Jingyuan and Ye, Zihe and Ju, Clark Mingxuan and Zhao, Tong and Shah, Neil and Chen, Li and Zhang, Yongfeng},
  year    = 2026,
  journal = {arXiv preprint arXiv:2601.08816},
}

@inproceedings{xu2025iagent,
  title     = {i{A}gent: {LLM} Agent as a Shield between User and Recommender Systems},
  author    = {Xu, Wujiang and Shi, Yunxiao and Liang, Zujie and Ning, Xuying and 
               Mei, Kai and Wang, Kun and Zhu, Xi and Xu, Min and Zhang, Yongfeng},
  booktitle = {Findings of ACL},
  pages     = {18056--18084},
  year      = {2025}
}

\clearpage

\appendix
\section{Full Proofs}
\addcontentsline{toc}{section}{Appendix A: Full Proofs}

\subsection{Upper Bound under Domain Shift and Group Imbalance}
\label{proof1}
\begin{proof}
Let $h = o \circ f_t$. By Lipschitz composition, $\mathrm{Lip}(h) \leq L_o L_f$. Consider probability measures $\mu, \nu$ on $(\Z, d_\Z)$. By the Kantorovich–Rubinstein duality for Wasserstein-1,
\[
|\E_{\mu}[h] - \E_{\nu}[h]| \leq \mathrm{Lip}(h) \, W_1(\mu, \nu).
\]
Choosing $\mu = \nu_t^0$ and $\nu = \nu_t^1$ yields
\[
\Gamma_{\mathrm{UGF}}^t = \sup_{o \in \mathcal{O}} \big| \E_{\nu_t^0}[h] - \E_{\nu_t^1}[h] \big| \leq L_o L_f \, W_1(\nu_t^0, \nu_t^1).
\]
Next, apply the triangle inequality for $W_1$:
\[
W_1(\nu_t^0, \nu_t^1) \leq W_1(\nu_t^0, \nu_s^0) + W_1(\nu_s^0, \nu_s^1) + W_1(\nu_s^1, \nu_t^1).
\]
Introduce $\varepsilon := W_1(\nu_t, \nu_s)$, $\delta_t^g := W_1(\nu_t^g, \nu_t)$, $\delta_s^g := W_1(\nu_s^g, \nu_s)$. By the triangle inequality,
\[
W_1(\nu_t^g, \nu_s^g) \leq W_1(\nu_t^g, \nu_t) + W_1(\nu_t, \nu_s) + W_1(\nu_s, \nu_s^g) = \delta_t^g + \varepsilon + \delta_s^g.
\]
Therefore,
\[
\Gamma_{\mathrm{UGF}}^t \leq L_o L_f \Big( W_1(\nu_s^0, \nu_s^1) + \delta_t^0 + \delta_t^1 + \delta_s^0 + \delta_s^1 + 2\varepsilon \Big).
\]
\end{proof}

\subsection{Uniform Convergence for Group Gaps}
\label{proof2}
\begin{proof}
Define $\mathcal{H} := \{ h = o \circ f_t \circ \phi : o \in \mathcal{O} \}$ and assume $\|h\|_\infty \leq B$. For group $g$, let $\{Z_i^{(g)}\}_{i=1}^{n_g}$ be i.i.d. from $\nu_t^g$. By standard symmetrization,
\[
\E\Big[ \sup_{h \in \mathcal{H}} \big| \E[h(Z)] - \tfrac{1}{n_g} \sum_{i=1}^{n_g} h(Z_i^{(g)}) \big| \Big] \leq 2 \, \mathfrak{R}_{n_g}(\mathcal{H}).
\]
Apply Hoeffding’s inequality for bounded functions to obtain, with probability at least $1-\delta$,
\[
\big| \E_{\nu_t^g}[h] - \tfrac{1}{n_g} \sum_{i=1}^{n_g} h(Z_i^{(g)}) \big| \leq 2 \, \mathfrak{R}_{n_g}(\mathcal{H}) + B \sqrt{\tfrac{\log(2/\delta)}{2 n_g}}.
\]
Apply this bound to $g_0$ and $g_1$ separately and use the triangle inequality to bound the difference of group means.
\end{proof}

\subsection{Fairness Preservation via Upper Bounds}
\label{proof3}
\begin{proof}
From A.1, $\Gamma_{\mathrm{UGF}}^t \leq L_o L_f ( W_1(\nu_s^0, \nu_s^1) + \delta_t^0 + \delta_t^1 + \delta_s^0 + \delta_s^1 + 2\varepsilon )$. If
\[
L_o L_f ( W_1(\nu_s^0, \nu_s^1) + \delta_t^0 + \delta_t^1 + \delta_s^0 + \delta_s^1 + 2\varepsilon ) \leq \Gamma_{\mathrm{UGF}}^0,
\]
then immediately $\Gamma_{\mathrm{UGF}}^t \leq \Gamma_{\mathrm{UGF}}^0$.
\end{proof}

\section{Time Complexity Analysis}
\label{app:timecomplexity}

In this section, we provide a concise analysis of the time complexity of our proposed \textsc{CDFA} framework. We decompose the complexity into two parts including the \textbf{base cross-domain recommender} and the \textbf{fairness-specific modules}.

\paragraph{Base CDR Module.}
Let $d$ be the hidden dimension of network embeddings, and $|U|$, $|I_s|$, $|I_t|$ the numbers of users, source-domain items, and target-domain items, respectively. For simplicity, we denote $|I| \approx |I_s| + |I_t|$, and $|U|$ includes overlapping and non-overlapping users.

When training a common cross-domain neural recommender model (e.g., graph-based), each epoch typically requires processing $O(|U|\times|I|\times d)$ to $O(|U|\times |I|^2\times d + |U|\times |I|\times d^2)$ operations, depending on the architectural specifics (e.g., GNN layers). For models adopting negative sampling, it is standard to sample a small number $k$ of unlabeled items per positive instance, yielding an $O(k)$ factor per user--item pair. Consequently, the base CDR component often exhibits complexity on par with classical single-domain sequential recommenders, typically $O(\mathcal{B}\times L)$ per epoch, where $\mathcal{B}$ is the mini-batch size and $L$ is a model-dependent factor such as $|I|^2d + |I|d^2$ in Transformer-based encoders~\cite{xu2025heterogeneous}. 

\paragraph{Fairness Modules in \textsc{CDFA}}
Our \textsc{CDFA} introduces two additional modules:

\begin{itemize}[leftmargin=1em]
    \item \textbf{Group-Aware Negative Sampler.} 
    We adaptively select unlabeled items for each group, guided by a momentum-based disparity measure. Specifically, in each iteration, we build a small candidate set $\mathcal{M}$ of size $|{M}| \ll |I|$ and compute sampling probabilities via a softmax over (user, item) scores. The sampling overhead is thus $O(|{M}|\times d)$ per user, which is typically minimal (e.g., $|{M}| \in \{1, 2, 4, 8,16\}$).
    
    \item \textbf{Cross-Domain Gain Redistribution.} 
    We estimate the information gain $\Delta \hat{I}_{st}$ once per mini-batch using three forward passes: (i)~$p(y|x_s)$; (ii)~$p(y|x_t)$; (iii)~$p(y|x_s,x_t)$, plus a small multi-layer perceptron (MLP) for the Estimator network (Eq.\,(21) in the main paper). The MLP dimension is of order $d$, so the overhead remains $O(\mathcal{B} \times d^2)$ per iteration. Finally, redistributing the gain among user groups simply adds a variance penalty term in the loss, incurring negligible cost.
\end{itemize}

\paragraph{Overall Complexity.}
Hence, per epoch, \textsc{CDFA} adds only modest overhead on top of the base CDR. The group-aware sampling scales as $O(|{M}| \times d)$ per instance, while the gain redistribution requires at most $O(\mathcal{B} \times d^2)$. Both are small relative to typical cross-domain recommendation complexities (often $O(\mathcal{B}\times L)$ with $L$ involving $|I|^2d$ or $|I|d^2$~\cite{xu2025heterogeneous,lin2025towards}). Consequently, our \textsc{CDFA} maintains a time complexity comparable to standard cross-domain methods, while significantly improving fairness. 

The worst-case cost for our Group-Aware Negative Sampling is \(O(|M| \cdot d)\) per user, scaling linearly even if \( |M| \) is increased for higher precision. The Gain Redistribution module operates at \(O(B \cdot d^2)\) per batch, independent of the total number of users (\(|U|\)) or edges (\(|E|\)). In contrast to graph-based baselines (e.g., BiTGCF), where complexity scales with graph density (\(O(|E|)\)), our approach preserves the scalability profile of the backbone model.

In practice, we observe no noticeable slowdowns during experiments, demonstrating that our \textsc{CDFA} is computationally efficient and scalable to large real-world scenarios.

\section{Experimental Setup}
\label{app:exp}

\subsection{Datasets}
We conduct experiments on two datasets widely used in the literature for evaluating CDR performance:  
\begin{itemize}  
    \item \textbf{Tenrec QB}~\cite{yuan2022tenrec}: Collected from the QQ Browser (QB) platform, this dataset includes user interactions from both the Video and Article domains.  
    \item \textbf{Tenrec QK}~\cite{yuan2022tenrec}: Sourced from the QQ Kan (QK) platform, this dataset also covers both Video and Article recommendation scenarios, and provides detailed demographic attributes.
\end{itemize}  
It should be noted that fairness-aware CDR methods require datasets containing both cross-domain user behaviors and sensitive attributes (e.g., \textit{gender} and \textit{age}), which limits the number of suitable datasets available. Table~\ref{tab:statistics} summarizes the statistics of the datasets. Following~\cite{tang2024fairness}, for both datasets, we use the video domain as the source and the article domain as the target, since the video domain has more data.

\begin{table}
\centering
\caption{Statistics of the cross-domain recommendation datasets used in our experiments, with their available sensitive attributes. }
\label{tab:statistics}
\begin{tabular}{llrrrc}
\toprule
\textbf{Dataset} & \textbf{Domain} & \textbf{\#Users} & \textbf{\#Items} & \textbf{\#Inter.} & \textbf{Attr.} \\
\midrule
\multirow{2}{*}{QB} & Video & 8,682 & 65,216 & 539,193 & gender, \\
 & Article & 21,639 & 1,859 & 93,832 & age\\
\midrule
\multirow{2}{*}{QK} & Video & 643,446 & 126,555 & 10,579,159 & gender, \\
 & Article & 31,544 & 14,285 & 476,678 & age \\
\bottomrule
\end{tabular}
\end{table}

\subsection{Baseline}
As our approach is model-agnostic, we validate its effectiveness by integrating CDFA into the following representative CDR methods and comparing performance before and after applying CDFA:
\begin{itemize}
    \item \textbf{CMF}~\cite{SG08} unifies the factorization of multiple rating matrices by sharing latent factors for entities that appear in more than one domain.
    \item \textbf{CLFM}~\cite{GLC13} utilizes a cluster-level latent factor scheme to capture both common cross-domain signals and domain-specific nuances.
    \item \textbf{BiTGCF}~\cite{LLL20} exploits a graph-based collaborative filtering framework to model higher-order relations, with overlapping users serving as bridges for information flow across domains.
    \item \textbf{CoNet}~\cite{HZY18} employs neural cross-connections to transfer knowledge effectively between domains.
    \item \textbf{DTCDR}~\cite{zhu2019dtcdr} adopts a dual-target strategy, fusing rating and content information via multi-task learning to enhance both domains concurrently.
\end{itemize}
Furthermore, we conduct comparisons between CDFA and several representative fairness-aware CDR methods:
\begin{itemize}
    \item \textbf{FairCDR}~\cite{tang2024fairness} quantifies the instance-level influence on both utility and fairness, and adaptively reweights training examples accordingly. This approach enables cross-domain knowledge transfer that alleviates distributional biases in sparse data and reduces group disparities.
    \item \textbf{VUG}~\cite{chen2025leave} synthesizes virtual users in the source domain, allowing all target-domain users to benefit from cross-domain transfer, thereby promoting fairness through disparity reduction.
\end{itemize}

\subsection{Implementation Details}
All experiments are conducted in PyTorch on a single NVIDIA Tesla V100 (32GB) GPU. To ensure reproducibility, we build upon the RecBole CDR framework~\cite{ZMH21}, adopting its standard preprocessing pipelines, dataset splits, and evaluation protocols~\cite{ZMH21, ZHP22}. Specifically, we employ an 8:2 train/validation split for the source domain and an 8:1:1 train/validation/test split for the target domain, with splits maintained on a per-user basis. For optimization, we use the Adam optimizer~\cite{KB14} with a default learning rate of 0.001 and a mini-batch size of 2048. The $L_2$ regularization coefficient is set to $10^{-4}$. The $\mathrm{Estimator}(\cdot,\cdot)$ component is implemented as a three-layer MLP with hidden layers of sizes 128, 64, and 64, respectively, employing ReLU activation and a dropout rate of 0.2. We tune the $\epsilon$ and $\gamma$ over the range $[0, 1]$, and vary the size of the unlabeled data candidate set $|M|$ from 1 to 16. As the momentum coefficient $\beta$ is not central to our method design, we fix its value at 0.9 in all experiments, following standard practice in momentum-based optimization and exponential moving average smoothing~\cite{sutskever2013importance, he2016deep}. For other baseline hyperparameters, we adopt the adjustment ranges specified in the respective original papers. All hyperparameters are tuned via grid search on the validation sets, and we report the best performance achieved for each baseline.

\end{document}